\documentclass[lettersize,journal]{IEEEtran}
\usepackage{amsmath,amsfonts}
\usepackage{array}
\usepackage[caption=false,font=footnotesize,labelfont=rm,textfont=rm]{subfig}
\usepackage{textcomp}
\usepackage{stfloats}
\usepackage{url}
\usepackage{verbatim}
\usepackage{graphicx}
\usepackage{cite}
\usepackage{multirow}
\usepackage{amssymb}
\usepackage{booktabs}
\usepackage[ruled,vlined,linesnumbered]{algorithm2e}
\SetKw{Continue}{continue}
\SetKwProg{Fn}{Procedure}{}{}
\newtheorem{theorem}{\bf Theorem}[section]
\SetKwInput{Input}{Input}
\SetKwInput{Output}{Output}
\newcolumntype{B}{>{\bfseries}c}
\begin{document}

\title{EviDC: A Violation-Guided Algorithm for Incremental Denial Constraint Discovery}

\author{Qian Zhou, Xixian Han, Xiaolong Wan 
\thanks{Qian Zhou, Xixian Han and Xiaolong Wan are with the School of Computer Science and Technology, Harbin Institute of Technology (HIT), China. (e-mail: 25S130328@stu.hit.edu.cn, hanxx@hit.edu.cn, wxl@hit.edu.cn).}
}


\maketitle

\begin{abstract}
Denial Constraints (DCs) are an important class of integrity constraints and have been widely used in data quality management. In dynamic datasets, newly inserted tuples may invalidate existing DCs and require the constraint set to be updated. Existing incremental DC discovery methods still generate a large amount of intermediate evidence because they do not exploit the structural information of existing DCs during evidence construction. We propose EviDC, a violation-guided incremental DC discovery algorithm. EviDC organizes existing DCs into a prefix tree structure called DCTrie, in which each path from the root to a leaf represents a potential violation path. During incremental processing, evidence is expanded only along reachable violation paths, while irrelevant branches are pruned as early as possible. We evaluate EviDC on real-world and synthetic datasets. The results show that EviDC reduces intermediate evidence and improves runtime efficiency in most scenarios. The performance gain becomes more pronounced as the insertion ratio and dataset size increase, showing the effectiveness and scalability of violation-guided evidence construction.
\end{abstract}


\section{Introduction}
\IEEEPARstart{D}{enial} constraints (DCs) are an important type of integrity constraints (ICs)~\cite{silberschatz2002database} in relational databases. They can express a wide range of complex semantic rules while remaining computationally manageable~\cite{chu2013discovering}. Due to their strong expressive power, DCs have been widely used in various data management tasks, including query optimization~\cite{kossmann2022data,pena2018mind}, data quality management~\cite{fan2007improving} and data cleaning~\cite{chu2013holistic,chu2016data}. 

A DC is defined over a set of predicates that cannot be satisfied simultaneously. A predicate is a comparison condition defined on a pair of tuples. For example, $t.ID=s.ID$ is a predicate. Taking the employee dataset De in Table~\ref{tab:employee} as an example, there are three rules that can be expressed as DCs. 

$\varphi_1:\lnot(t.ID=s.ID)$, describing the semantics that there are no employees with the same ID in De; 

$\varphi_2:\lnot(t.Level=s.Level \land t.Mgr \ne s.Mgr)$, meaning that if two employees have the same level, then their supervisors should be the same; 

$\varphi_3:\lnot(t.Hired<s.Hired \land t.Level<s.Level)$, meaning that if one employee's start date is earlier than another employee's, then their level should not be lower than the latter's.

\begin{table}
\centering
\caption{Employee Dataset De}
\label{tab:employee}
\begin{tabular}{c|ccccc}
\hline
 & ID & Name & Hired & Level & Mgr \\
\hline
$t_1$ & \#1 & Ana & 2000 & 5 & \#1 \\
$t_2$ & \#2 & Sam & 2001 & 4 & \#1 \\
$t_3$ & \#3 & Ana & 2002 & 2 & \#2 \\
$t_4$ & \#4 & Kai & 2002 & 2 & \#2 \\
$t_5$ & \#5 & Tom & 2003 & 1 & \#3 \\
\hline
\end{tabular}
\end{table}

DCs generalize several traditional integrity constraints, including unique column combinations (UCCs)~\cite{abedjan2014detecting}, functional dependencies (FDs)~\cite{huhtala1999tane}, order dependencies (ODs)~\cite{szlichta2012fundamentals} and so on. Using these traditional expressions, the above three rules require three different constraint types. $\varphi_1$ can be represented by key constraints, $\varphi_2$ is a typical example of functional dependencies, and $\varphi_3$ can be represented by order dependencies~\cite{abedjan2019data}. In addition, the importance of DCs among different forms of ICs also comes from their good tractability~\cite{chu2013discovering}.

Due to the significant application value of DCs, the DC discovery problem has been proposed to automatically discover DCs from datasets, thereby reducing the cost of manual design. Most existing DC discovery algorithms follow a two-phase framework consisting of evidence construction and DC enumeration~\cite{bian2024discovering}. In this framework, evidence records the predicates satisfied by tuple pairs and serves as the core intermediate structure. Based on this framework, several static DC discovery algorithms have been proposed, such as Hydra~\cite{bleifuss2017efficient} and DCFinder~\cite{pena2019discovery}.

In practical application scenarios, data are not static but are frequently updated through insertions. For example, e-commerce systems continuously receive new transaction records~\cite{jain2021overview}, and industrial systems continuously collect equipment operation logs~\cite{tamalu2023fault}. Suppose a newly hired employee in De is $t_6=(\#1,jack,2001,5,\#1)$. It has the same $ID$ value as $t_1$, which violates the requirement of $\varphi_1$. At this time, $\varphi_1$ is invalid. If the DCs cannot be updated and maintained in time, valid newly inserted employee records may be incorrectly flagged as violations.

A straightforward solution is to recompute the DCs in dynamic datasets by rerunning static algorithms. However, this strategy introduces a large amount of redundant computation and is costly~\cite{tan2020fast}. Accordingly, the problem of incremental DC discovery has been proposed, aiming to recognize invalid DCs and update them to resolve the violations.

To address the incremental DC discovery problem, IncDC~\cite{qian2023incremental} and 3DC~\cite{pena2024discovering} have been proposed. Although they avoid rediscovering DCs from scratch, both methods still generate a large amount of redundant evidence. They first construct incremental evidence for affected tuple pairs and then use the generated evidence to identify violated DCs and update the constraint set.

The main difference lies in how the evidence is obtained and maintained. IncDC relies on index structures to filter candidate tuple pairs, which introduces considerable memory overhead and frequent index maintenance. 3DC maintains the evidence set produced in the previous discovery round to identify newly appearing evidence, and this evidence set can be costly to construct and maintain. As a result, many intermediate evidence records may be generated or stored even though they will never lead to a DC violation.

This observation motivates a different direction. Instead of constructing evidence first and checking DCs afterward, we use the existing DCs to guide evidence construction from the beginning. Each DC already defines a potential violation path. A tuple pair violates a DC only if it satisfies all predicates on that path. For example, $\varphi_2$ corresponds to such a violation path. To violate $\varphi_2$, a tuple pair must satisfy both $t.Level=s.Level$ and $t.Mgr \ne s.Mgr$. Therefore, once the partial evidence of a tuple pair can no longer be extended to any existing DC, further evidence expansion for this tuple pair is unnecessary.

Based on this observation, we propose EviDC, a violation-guided incremental DC discovery method. We organize all potential violation paths using the existing DCs. Once a tuple pair deviates from every possible violation path, its subsequent expansion becomes unnecessary and can be pruned immediately. In this way, EviDC differs from IncDC and 3DC by embedding DC structural information directly into the evidence construction process, rather than detecting invalid DCs only after evidence has been generated.

To support violation-guided evidence construction, we design a prefix-sharing structure called DCTrie to organize potential violation paths implied by existing DCs. In DCTrie, each node corresponds to a predicate, and each root-to-leaf path represents a DC. During incremental processing, evidence is expanded only along reachable paths in DCTrie. Once a tuple pair deviates from all possible violation paths, the corresponding evidence construction process terminates immediately. In this way, DCTrie effectively reduces the generation of redundant intermediate evidence.

EviDC follows a two-stage framework consisting of incremental evidence construction and DC updating. EviDC focuses on reducing evidence construction cost, and uses the repair strategy adopted by existing incremental methods for DC update.

The main contributions of this paper can be summarized as follows:

(1) We propose EviDC, a violation-guided incremental DC discovery algorithm. EviDC incorporates the structural information of existing DCs into incremental evidence construction, so that evidence is expanded only when it may still lead to a DC violation.

(2) We design DCTrie, a prefix-sharing structure that organizes existing DCs as potential violation paths. DCTrie supports reachability checking during evidence construction and enables early pruning of irrelevant evidence contexts.

(3) We conduct extensive experiments on real-world and synthetic datasets. Results show that EviDC significantly reduces evidence size and achieves better efficiency in most scenarios, especially under large insertion ratios.
        
The remainder of this paper is organized as follows. Section~\ref{section:related_work} reviews related work on DC discovery and incremental DC discovery. Section~\ref{section:preliminaries} presents the preliminaries and problem definition. Section~\ref{section:algorithm} introduces the EviDC algorithm. Section~\ref{section:experiment} presents the experimental evaluation. Finally, Section~\ref{section:conclusion} concludes the paper.

\section{Related Work} \label{section:related_work}
In recent years, DC discovery has become an important research direction in data quality management. Based on different application scenarios, existing research can be divided into two categories: static algorithms and incremental algorithms.

Early studies on DC discovery mainly focused on static datasets. Chu et al.~\cite{chu2013discovering} first formalized denial constraints and proposed FASTDC, which established the widely adopted two-phase framework consisting of evidence construction and DC enumeration. 

Several studies focus on improving evidence construction efficiency. BFASTDC~\cite{pena2018bfastdc} optimized predicate verification through bitwise operations. DCFinder~\cite{pena2019discovery} introduced Position List Indexes (PLIs) to accelerate evidence construction. Pena et al.~\cite{pena2022fast} further introduced parallel evidence pipelines to improve evidence generation efficiency. FASTADC~\cite{xiao2022fast} reduced evidence construction cost using clue sets, while Filho et al.~\cite{marques2023discovering} applied Boolean operations to reduce the memory overhead of intermediate evidence structures.

Other studies mainly optimize the DC enumeration stage. Hydra~\cite{bleifuss2017efficient} proposed an evidence inversion strategy to reduce enumeration overhead, while ADCMiner~\cite{livshits2020approximate} adopted the state-of-the-art hitting set enumeration algorithm MMCS~\cite{murakami2013efficient} to improve DC enumeration efficiency.

In addition, there are studies that break through this framework. DCMiner~\cite{bian2024discovering} mines concise, diverse, and strongly correlated DCs based on deep reinforcement learning. DCRer~\cite{wu2025deep} also discovers DCs based on deep reinforcement learning and combines them with downstream data cleaning tasks. DCArray~\cite{marques2026discovery} performs DC discovery in a hardware pipeline based on FPGA, which significantly improves execution efficiency.

Although these methods improve different stages of static DC discovery, most of them still rely on constructing large intermediate evidence structures, which remains the fundamental bottleneck of DC discovery algorithms.

Different from static algorithms, incremental DC discovery avoids the complexity in global computation by processing the updated parts of the data. This incremental problem has been studied in other constraint types, including UCCs~\cite{abedjan2014detecting}, FDs~\cite{schirmer2019dynfd,xiao2022dynamic,caruccio2019incremental} and so on. In DC discovery, current incremental methods include IncDC~\cite{qian2023incremental} and 3DC~\cite{pena2024discovering}. Overall, they can be viewed as following a two-stage framework: first constructing incremental evidence, and then performing DC verification and update.

IncDC is the first method for incremental DC discovery. In the evidence construction stage, IncDC builds novel index structures from the original data and existing DCs, especially over highly selective predicate combinations, to filter candidate tuple pairs and accelerate incremental evidence generation. In the DC update stage, the generated evidence is used to identify invalid DCs and derive repaired constraints by extending violated DCs with additional predicates, followed by validity and redundancy checking. Although IncDC avoids rediscovering DCs from scratch, its efficiency heavily depends on the maintenance of complex index structures.

3DC is an incremental method built on an existing DC discovery algorithm~\cite{pena2022fast}. It introduces a parallel working pipeline based on evidence context. 3DC reuses the evidence set $E_r$ obtained from the previous discovery round. For tuple insertions, it first constructs the incremental evidence set $E_{\Delta r}$, derives newly appearing evidence by computing $E^{inc}=E_{\Delta r}\setminus E_r$, and then updates the DCs through dynamic DC enumeration. Under this framework, the maintenance of DCs is transformed into a matching problem between incremental evidence and the existing DCs. Unlike IncDC, 3DC adopts separate processing strategies for insertions and deletions, addressing IncDC's limited support for deletion scenarios. However, 3DC relies on the availability of the complete evidence set $E_r$, whose construction and maintenance can be costly.

Our work differs in the following. Both IncDC and 3DC still construct evidence independently from the structural information of existing DCs. We organize existing DCs as potential violation paths. We also directly incorporate DC structural information into the evidence construction stage. EviDC constrains evidence expansion to paths that may still lead to violations, rather than generating evidence first and verifying violations afterward.

\section{Preliminaries} \label{section:preliminaries}
Let $R(A,B,C,\ldots)$ denote a relational schema, where $A,B,C$ are attributes in the relation. $r$ denotes an instance of $R$, representing a set of tuples in this schema. Let $|R|$ denote the number of attributes in $R$. For any tuple $t\in r$, $t.A$ is the projection of tuple $t$ on attribute $A$.

{\bf{Predicates: }}predicates describe the comparison relationship between two tuples over a given attribute or a comparison between a tuple's attribute value and a constant. For any two tuples $t,s \in r$, a predicate can be defined as:
$$P(t,s): t.A \ op \ s.B\quad or \quad t.A \ op\ c$$

where $op \in \{<,\leq,>,\geq,=,\ne\}$, $c$ is a constant. For example, $P:t.name=s.name$ is a predicate, indicating that two employees have the same name. This paper focuses on comparisons between two different tuples, considering only the case of $t\ne s$.

{\bf{Denial constraints: }}a denial constraint is the negation form of a set of predicate conjunctions. Given a set of predicates, DC can be defined as:
$$\varphi: \forall t, s \in r, \neg (P_1(t, s) \land P_2(t, s) \land \dots \land P_k(t, s))$$

In other words, a DC is a predicate set that cannot be satisfied simultaneously. If a DC does not contain any predicates with constant values, it is called a variable denial constraint (VDC). Otherwise, it is called a constant denial constraint (CDC). If a DC $\varphi$ holds on $r$, we write $\varphi \models r$.

Based on the above definitions, we further discuss several basic properties of DCs.
(1) Triviality: A DC is said to be trivial if it is satisfied by any instance. For example, $\lnot(t.ID=s.ID \land t.ID \ne s.ID)$ contains two predicates that can never be true simultaneously. In this case, it loses its semantics, so we only consider nontrivial DCs. 

(2) Symmetry: For a DC $\varphi = \neg (P_1(t, s) \land P_2(t, s) \land \dots \land P_k(t, s))$, its symmetric form is defined as $\varphi^{sym} = \neg (P_1(s,t) \land P_2(s,t) \land \dots \land P_k(s,t))$. A symmetric DC can be obtained by swapping the order of tuple pairs in a DC. The two are semantically equivalent, except that the order of the tuple pairs has been swapped. The symmetric DC of $\varphi_3$ is $\varphi_4:\lnot(t.Hired>s.Hired \land t.Level>s.Level)$.

(3) Augmentation: If $\lnot(P_1\land\ldots\land P_n)$ is valid, then $\lnot(P_1\land\ldots\land P_n\land Q)$ is also valid.

This means that if a valid DC is extended by adding a predicate, it remains valid. For example, if $\varphi_1:\lnot(t.ID = s.ID)$ holds, then another DC, $\lnot(t.ID = s.ID\land t.Name=s.Name)$ also holds. Therefore, the objective of DC discovery is to find minimal DCs.

{\bf{Predicate Space: }}Given a relational schema $R$, its predicate space is denoted by $P$, representing all possible predicates. For categorical data, attributes can only be compared using equality or inequality operators, i.e., $op\in\{=,\ne\}$, and two different predicates can be generated for the same attribute. For numerical data, $op\in\{=,\ne,>,\geq,<,\leq\}$, and six predicates can be generated for the same attribute.

Taking the employee dataset as an example, its corresponding predicate space is shown in Table~\ref{tab2}.
\begin{table}[!t]
\centering
\caption{Employee Predicate Space}
\label{tab2}
\resizebox{\columnwidth}{!}{
\begin{tabular}{lll}
\hline
$P_1: t.ID = s.ID$        & $P_2: t.ID \ne s.ID$       & $P_3: t.Name = s.Name$ \\
$P_4: t.Name \ne s.Name$  & $P_5: t.Hired < s.Hired$   & $P_6: t.Hired \le s.Hired$ \\
$P_7: t.Hired > s.Hired$  & $P_8: t.Hired \ge s.Hired$ & $P_9: t.Hired = s.Hired$ \\
$P_{10}: t.Hired \ne s.Hired$ & $P_{11}: t.Level < s.Level$ & $P_{12}: t.Level \le s.Level$ \\
$P_{13}: t.Level > s.Level$   & $P_{14}: t.Level \ge s.Level$ & $P_{15}: t.Level = s.Level$ \\
$P_{16}: t.Level \ne s.Level$ & $P_{17}: t.Mgr = s.Mgr$       & $P_{18}: t.Mgr \ne s.Mgr$ \\
\hline
\end{tabular}
}
\end{table}

{\bf{Evidence: }}The concept of evidence is proposed for recording how tuple pairs satisfy predicates in the predicate space. For any two distinct tuples $(t,s)\in r$, the evidence corresponding to this tuple pair is defined as:
$$Evi(t,s)=\{p\in P(R)\mid p(t,s)=true\}$$
Thus, evidence is essentially the set of all predicates that are satisfied by a tuple pair. The evidence set over the whole dataset is defined as the collection of evidence generated by all tuple pairs:
$$E(r)=\{Evi(t,s)| t,s\in r,t\neq s\}$$

Based on evidence, the validity of a DC can be defined as follows:

A DC $\varphi=\lnot(P_1\land P_2\land \cdots \land P_k)$ is valid if $\forall Evi\in E(r),P_1,\ldots,P_k \nsubseteq Evi$

A DC $\varphi$ holds on a relational instance if and only if there does not exist any evidence that contains all predicates in $\varphi$. Conversely, if some evidence contains all predicates in $\varphi$, then this evidence violates the DC. For example, $(t_6,t_1)$ in De corresponds to the evidence $E_1=\{P_1,P_4,P_5,P_6,P_{10},P_{12},P_{14},P_{15},P_{17}\}$. $\varphi_1$ is a subset of $E_1$, that is, all predicates in $\varphi_1$ are contained in $E_1$. We can say that the tuple pair $(t_6,t_1)$ violates $\varphi_1$.

Based on the above core concepts, we next give a formal description of the DC discovery problem.

{\bf{Static DC discovery: }} Given a relational instance $r$, the goal is to find the set $\Sigma$ of DCs that hold on $r$.

{\bf{Incremental DC discovery: }}
Given an initial dataset $r$, an original DC set $\Sigma$, and inserted tuples $\Delta r$, the goal is to identify the set of invalid DCs $\Delta\Sigma^{-}$ and incrementally repair them to maintain the validity of the DCs.

In incremental DC discovery, incremental evidence is constructed for affected tuple pairs. These affected tuple pairs are denoted as $(t,t')$, where at least one tuple belongs to $\Delta r$. If some incremental evidence covers a DC $\varphi$, then $\varphi$ becomes invalid and should be removed.

\section{EviDC Algorithm} \label{section:algorithm}

\subsection{Overview of Algorithm}
Incremental DC discovery aims to efficiently maintain the validity of DCs under data updates. We propose EviDC, which applies a violation-guided evidence construction method that directly constrains evidence expansion using the structural information of existing DCs. We regard each DC as a potential violation path composed of a sequence of predicates. Based on this, the algorithm can more efficiently identify affected DCs while reducing the generation of irrelevant evidence.

To support this process, we design a tree structure called DCTrie, which compactly organizes shared predicates and potential violation paths. During incremental processing, evidence is expanded only along reachable paths in DCTrie, while branches that can no longer lead to violations are pruned immediately.

As shown in Figure~\ref{fig:overview}, EviDC can be divided into three main parts: DCTrie construction, violation-guided evidence construction, and DC update.

\begin{figure}
\centering
\includegraphics[width=3.4in]{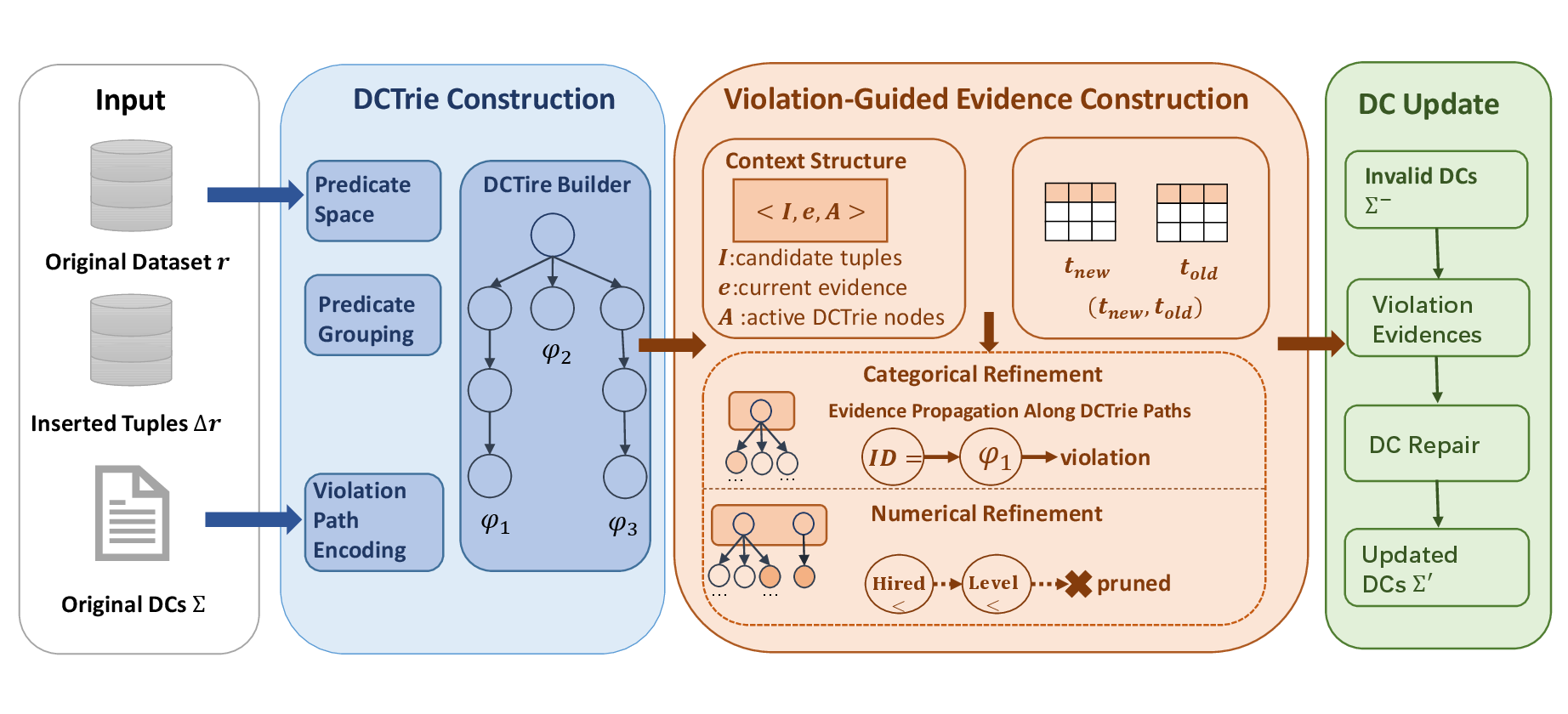}
\caption{Overview of EviDC and its components.}
\label{fig:overview}
\end{figure}

Algorithm~\ref{alg:evidc} presents the overall process of EviDC. First, the algorithm constructs the predicate space and divides the predicates into groups. Then, it builds the DCTrie and constructs the violation evidence set and the affected DCs. Finally, the algorithm removes the detected invalid DCs and generates repaired DC candidates to eliminate violations.

\begin{algorithm}[!t]
\caption{\textsc{EviDC}}
\label{alg:evidc}

\Input{A relational database instance $r$, inserted tuples $\Delta r$, and an existing DC set $\Sigma$}
\Output{The maintained DC set $\Sigma'$}

Construct the predicate space $P$ for $r$\;
Divide the categorical groups $G_{cat}$ and numerical groups $G_{num}$\;
Build predicate indexes\;
Sort $G_{cat}$ and $G_{num}$\;
$G \leftarrow \langle G_{cat},G_{num}\rangle$\;
$T \leftarrow \textsc{DCTrieBuilder}(\Sigma,P,G)$\;
$(E_{\Delta},\Sigma^{-}) \leftarrow \textsc{IncreEviConstruction}(r,\Delta r,T,P,G)$\;

\If{$\Sigma^{-}=\emptyset$}{
    \Return{$\Sigma$}\;
}

\ForEach{$\varphi\in\Sigma^{-}$}{

    $E_{\varphi}\leftarrow$ evidence subset that violates $\varphi$\;

    $P_{add}\leftarrow P \setminus$ all predicates that appear in $E_{\varphi}$\;

    \ForEach{$p\in P_{add}$}{

        Generate a new DC $\varphi'\leftarrow\varphi\cup\{p\}$\;

        \If{$\varphi'$ is valid and non-trivial}{

            Add $\varphi'$ into $\Sigma_r$\;
        }
    }
}

$\Sigma'\leftarrow(\Sigma\setminus\Sigma^{-})\cup\Sigma_r$\;
\Return{$\Sigma'$}\;

\end{algorithm}

\subsection{DCTrie Structure}
This paper adopts a prefix tree DCTrie to uniformly represent all existing DCs. The basic idea of DCTrie is to organize the predicates in DCs into tree nodes according to a fixed attribute order. It merges and represents the constraint paths that have the same prefix. Each DC is mapped to a path from the root node to the leaf node to correspond to the violation path.

To ensure that predicates over the same attribute appear in the same layer, the predicate space needs to be grouped first. For attribute $A$, define its corresponding predicate group $G_A \subseteq P$ as the set of all predicates acting on attribute $A$, where $P$ is the predicate space. All predicate groups form a partition $G = G_1,\ldots, G_m$. Predicates in the same group are mutually exclusive. They are divided into different nodes in a tree's layer. For categorical attributes, the same attribute generates only equal and unequal predicates, such as $P_1: t.ID = s.ID$ and $P_2: t.ID \ne s.ID$. A tuple pair can only satisfy one of them, so they correspond to different nodes in the same layer. For numerical attributes, since there are implications among the six comparison operations, they cannot be regarded as six completely independent directions. Based on the conditions that are simultaneously satisfied, they are divided into three directions: $\{=,\ge,\le\}$, $\{>,\ge,\ne\}$, and $\{<,\le,\ne\}$. Each layer of the tree represents a selection branch on the attribute group. Different nodes in the same layer correspond to different predicates in the attribute group.
 
After predicate grouping, it is also necessary to determine the order of attribute groups in DCTrie. This order directly affects the pruning effect. To eliminate the tuple pairs that are unlikely to lead to violations as early as possible, we prioritize placing the attribute groups with stronger discrimination ability at shallower layers. We sort the attribute groups based on the diversity of attribute values and prioritize processing categorical attributes, then numerical attributes. This is because categorical attributes usually correspond to only a few discrete predicate branches, and their comparison results are easier to obtain. If early branches are available, the algorithm can avoid excessive branching in numerical attributes. For the same type of attribute groups, we further determine the order based on the discrimination ability of attribute values. 

Based on the above design, we formally define DCTrie as follows:
$$T=(V,E,root)$$

where $V$ represents the nodes, $E$ represents the edge set, $root$ represents the empty predicate node. DCTrie satisfies the following properties:
\begin{enumerate}
\item{The root node does not correspond to any specific predicate, and its semantics is the empty prefix.}
\item{Except for the root node, each node corresponds to a predicate or a set of predicates that can all be satisfied simultaneously.}
\item{Each layer of the tree corresponds to a predicate group. Each leaf node corresponds to a DC.}
\item{If a DC does not contain any predicates in the attribute group $G_i$, then a wildcard node "*" is introduced at the i-th layer of the DCTrie. This special node is called a wildcard node (WildcardChild) and is used to represent that the current DC does not impose any constraints on the attribute group at that layer.}
\item{Leaf nodes store the DCs corresponding to their paths, which are used for DC update.}
\end{enumerate}

In DCTrie, any path from the root node to a leaf node is regarded as a violation path. If all predicates on this path are satisfied, then the DC corresponding to this path is violated.

To ensure the completeness of the matching of the violation paths, DCTrie also needs to consider the symmetric form of a DC during its construction. For a DC $\varphi \in \Sigma$, its symmetric form $\varphi^{sym}$ can be obtained by exchanging the tuple order. The same violation may be caused by a tuple pair $(t,s)$ or $(s,t)$. They are the same tuple pair but have different evidence. To avoid missing the violation caused by tuple swapping, when constructing DCTrie, both the DC and its symmetric form need to be inserted into the tree simultaneously.

Algorithm~\ref{alg:dctrie_builder} presents the construction process of DCTrie. The insertion process of a single DC is recursively executed in predicate group order. For the predicate group corresponding to the current layer, if the DC contains a predicate in this group, it expands downward along the corresponding predicate node. If it does not contain a predicate, it recursively follows the wildcard node until all predicate groups are processed and records the DC at the leaf node.

\begin{algorithm}[!t]
\caption{\textsc{DCTrieBuilder}}
\label{alg:dctrie_builder}

\Input{$\Sigma$, predicate space $P$, sorted attribute group $G=\langle G_1,\ldots,G_m\rangle$}

\Output{DCTrie $T$}

Create the root node of $T$\;

\ForEach{$\varphi\in\Sigma$}{

    \textnormal{\textsc{InsertDC}}$(root,\varphi,G,1)$\;

    \If{$\varphi^{sym}\neq\varphi$}{

        \textnormal{\textsc{InsertDC}}$(root,\varphi^{sym},G,1)$\;
    }
}

\Return{$T$}\;

\vspace{1mm}

\Fn{\textnormal{\textsc{InsertDC}}$(node,\varphi,G,i)$}{

    \If{$i>|G|$}{

        mark $node$ as a leaf node and associate with $\varphi$\;

        \Return\;
    }

    $S\leftarrow$ the set of predicates of $\varphi$ belonging to $G_i$\;

    \If{$S=\emptyset$}{

        $child\leftarrow$ wildcard child of $node$\;

        \textnormal{\textsc{InsertDC}}$(child,\varphi,G,i+1)$\;
    }

    \Else{

        \ForEach{$p\in S$}{

            $child\leftarrow$ child node of $node$ labeled by $p$\;

            \textnormal{\textsc{InsertDC}}$(child,\varphi,G,i+1)$\;
        }
    }

    \Return\;
}

\end{algorithm}

For the De dataset, we build the DCTrie using the attribute order $ID$, $Name$, $Mgr$, $Hired$, and $Level$. The resulting structure is shown in Figure~\ref{fig:dctrie}. Here, the wildcard node $*$ indicates that the corresponding DC has no predicate constraints on the attribute group at this layer. To simplify the notation, we use $A=$ to denote the predicate $t.A=s.A$. The green nodes represent the paths formed by the symmetric insertion of DCs. 

\begin{figure}[!t]
\centering
\includegraphics[width=3.3in]{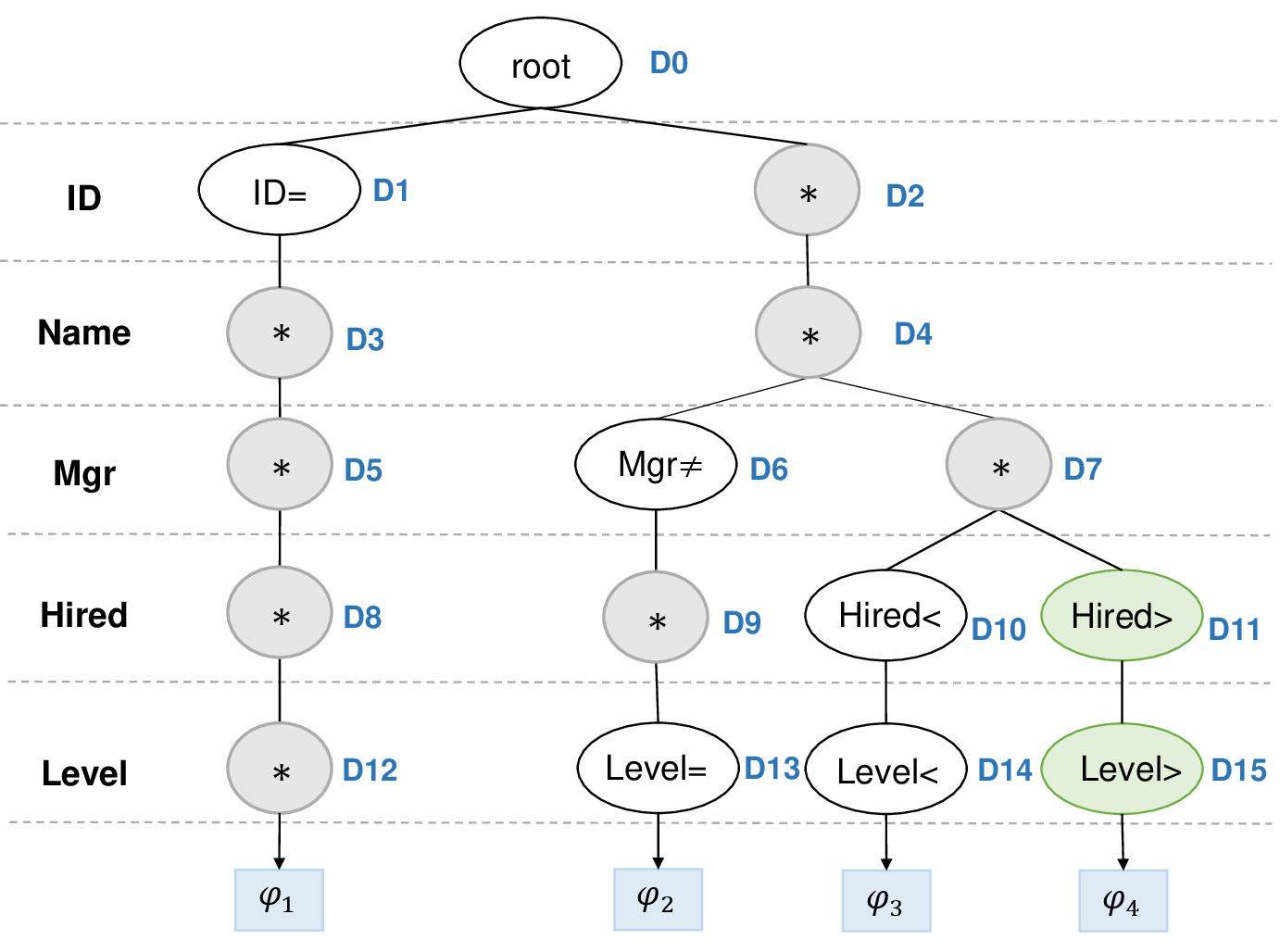}
\caption{Structure of DCTrie.}
\label{fig:dctrie}
\end{figure}

\subsection{Incremental Evidence Construction Based on DCTrie}
After building DCTrie, we utilize it to guide the evidence construction. When handling new tuples, the predicate refinement process is directly advanced along the valid path in DCTrie and dynamically prunes redundant branches. Thus, the algorithm does not need to explicitly generate a large amount of irrelevant evidence, thereby reducing the size of intermediate results.

In order to merge the same evidence as much as possible, we introduce a structure called Evidence Context~\cite{pena2022fast} to store the corresponding evidence for tuple pairs. It is an intermediate data structure, recording the evidence and relevant context information. We also maintain an active node set in each evidence context to provide pruning information. For an inserted tuple $t_{new}$, the evidence context is defined as follows:
{\small
$$Evidence\ Context=\{t_{new},RightTuples,Evidence,Active\}$$
}

Here, $RightTuples$ denotes the set of tuples that have the same evidence as $t_{new}$. $Evidence$ records the predicates jointly satisfied by these tuple pairs. $Active$ denotes the set of currently reachable DCTrie nodes. This serves as the basic unit in evidence construction. An evidence context can be viewed as a group of tuple pairs that have not yet been pruned. These tuple pairs share the same partial evidence and path state at a certain stage.

For each inserted tuple, the algorithm first builds an initial context. In the initial context, $RightTuples$ is initialized as $r\cup\Delta r\setminus\{t_{new}\}$, representing all tuples that may form tuple pairs with $t_{new}$. $Evidence$ is initialized using the default predicate directions of each attribute group, which may not actually be satisfied. The $Active$ set is initialized as the root node.

After context initialization, the algorithm progressively refines evidence contexts following the hierarchical order of DCTrie. At the $i$-th layer, only predicates belonging to the current predicate group $G_i$ are considered. Tuple pairs are divided into different sub-contexts according to the predicates they satisfy in the current group, which is referred to as context splitting. The process of updating the current evidence according to the satisfied predicates is called refinement. After refinement, tuple pairs in different sub-contexts still share the same evidence before the current layer, but correspond to different predicate selections at the current layer.

During refinement, the Active set is propagated simultaneously to maintain all reachable DCTrie paths. For a refinement predicate $p$, the algorithm retains two types of reachable nodes: child nodes labeled by p and inherited wildcard nodes. The resulting Active set therefore contains all DCTrie nodes that remain reachable after the current refinement step.

If a context's $Active$ set is empty after refinement, it indicates that this context cannot match any complete DCTrie path. Therefore, the tuple pairs from the context cannot violate any DCs. At this point, the subsequent refinement process of this context is terminated without further verification of the remaining predicates.

The order of attribute groups not only affects the structural organization of DCTrie, but also directly influences the efficiency of refining the evidence context. For an inserted tuple, its initial context usually corresponds to a larger set of right tuples. If we prioritize the more discriminative attribute group, then this right tuples set can shrink rapidly in the early stage, thereby reducing the subsequent splitting of attribute groups. Taking the insertion of data $t_6$ as an example, the initial right tuples are $\{t_1,t_2,t_3,t_4,t_5\}$. First, the refinement is carried out at the ID layer, and the tuples that satisfy $t.ID = s.ID$ are only $t_1$. The remaining tuples fall into another sub-context. At this time, for the sub-context containing only $t_1$, since its size has significantly shrunk, no further splitting occurs in the subsequent layers and no comparison operations are required.

To avoid excessive fragmentation of evidence contexts, EviDC initializes evidence using the default predicate direction of each attribute group. These directions make up the initial evidence. For categorical attributes, initially the default direction is set to be unequal. For numerical attributes, predicates that satisfy $\{>,\ge,\ne \}$ are adopted. The purpose is to enable most tuple pairs to continue propagating along a unified default path, while only a small subset of tuple pairs needs to be corrected. For the $ID$ attribute, the tuple pairs satisfying $P_1:t.ID = s.ID$ are only $(t_6,t_1)$, so only one tuple pair needs to be corrected to an equality predicate.

Using the inserted tuples in Table~\ref{tab:employee_inserted} as an example, the initial evidence context constructed for $t_8$ is: $C_0(t_8)=\big\langle\,\{t_1,t_2,t_3,t_4,t_5,t_6,t_7\},\{P_2,P_4,P_7,P_8,\ldots,P_{18}\},\{\mathrm{root}\}\big\rangle$.

\begin{table}[!t]
\centering
\caption{Newly Inserted Tuples in Employee Dataset}
\label{tab:employee_inserted}
\begin{tabular}{c|ccccc}
\hline
 & ID & Name & Hired & Level & Mgr \\
\hline
$t_6$ & \#1 & Jack & 2001 & 5 & \#1 \\
$t_7$ & \#6 & Ema  & 2003 & 3 & \#2 \\
$t_8$ & \#7 & Wile & 2002 & 1 & \#2 \\
\hline
\end{tabular}
\end{table}

Algorithm~\ref{alg:incremental_evidence_construction} presents the overall incremental evidence construction procedure. For each newly inserted tuple, the algorithm first initializes a context set, then sequentially calls Algorithm~\ref{alg:update_categorical_context} and Algorithm~\ref{alg:update_numerical_context} to refine the categorical attribute group and the numerical attribute group, respectively. After all attribute groups have been processed, it collects the final evidence and the invalid DCs.

\begin{algorithm}[!t]
\caption{IncreEviConstruction}
\label{alg:incremental_evidence_construction}

\Input{$r$, $\Delta r$, DCTrie $T$, $P$, predicate groups $G$}

\Output{Incremental evidence set $E_{\Delta}$, invalid DCs $\Sigma^{-}$}

Initialize $E_{\Delta}\leftarrow \emptyset$ and $\Sigma^{-}\leftarrow \emptyset$\;

Construct the base evidence $Evi$\;

\ForEach{inserted tuple $t_i\in\Delta r$}{
    $I_r \leftarrow r \cup (\Delta r \setminus {t_i})$\;

    $C\leftarrow\{\langle I_r,Evi,\{T.root\}\rangle\}$\;

    \ForEach{categorical group $g\in G_{cat}$}{
    
        $C\leftarrow UpdateCategoricalContext(C,t_i,g,T)$\;
    }

    \ForEach{numerical group $g\in G_{num}$}{
    
        $C\leftarrow UpdateNumericalContext(C,t_i,g,T)$\;
    }

    $(E_i,\Sigma_i^{-})\leftarrow CollectEvidence(C,T)$\;

    $E_{\Delta}\leftarrow E_{\Delta}\cup E_i$\;

    $\Sigma^{-}\leftarrow \Sigma^{-}\cup \Sigma_i^{-}$\;
}

\Return{$(E_{\Delta},\Sigma^{-})$}\;

\end{algorithm}

We now describe in detail the context refinement and pruning on categorical attribute groups. For a categorical attribute group, the same attribute involves both equality predicates and inequality predicates. Since we initially assume that a tuple pair satisfies the inequality predicate, tuples that satisfy the equality relationship need to be removed from the right tuple set. At this point, we perform context splitting: one sub-context corresponds to tuple pairs that satisfy the equality predicate, and the other corresponds to tuple pairs that remain on the inequality branch. For the former, we need to update the evidence by replacing the inequality predicate with the equality predicate, and then advance the active nodes along the equality branch. For the latter, we keep the original evidence and only need to update the set of active nodes. If the active node set becomes empty, it means that the current context cannot possibly form a violation, and should be pruned.

\begin{algorithm}[!t]
\caption{\textsc{UpdateCategoricalContext}}
\label{alg:update_categorical_context}

\Input{The current context set $C$, inserted tuple $t_i$, categorical predicate group $g$, and DCTrie $T$}

\Output{Updated context set $C'$}

$p_{eq},p_{neq}\leftarrow$ the equality and inequality predicates in $g$\;

$I_{eq}\leftarrow$ the tuple ID set with the same value as $t_i$ on the attribute of $g$\;

\If{$I_{eq}=\emptyset$}{ \label{line:nullI_eq}

    Initialize $C'\leftarrow\emptyset$\;

    \ForEach{$\langle I,e,A\rangle\in C$}{

        $A'\leftarrow T.\textsc{Advance}(A,p_{neq})$\;

        \If{$A'\neq\emptyset$}{

            Add $\langle I,e,A'\rangle$ into $C'$\;
        }
    }

    \Return{$C'$}\; \label{line:handleI_eq}
}

Initialize $C_{add}\leftarrow\emptyset$ and $C_{keep}\leftarrow\emptyset$\;

\ForEach{$\langle I,e,A\rangle\in C$}{

    $I_{match}\leftarrow I\cap I_{eq}$\;

    $I_{rem}\leftarrow I\setminus I_{match}$\;

    $A_{eq}\leftarrow T.\textsc{Advance}(A,p_{eq})$\;

    $A_{neq}\leftarrow T.\textsc{Advance}(A,p_{neq})$\;

    \If{$A_{neq}\neq\emptyset$ and $I_{rem}\neq\emptyset$}{

        Add $\langle I_{rem},e,A_{neq}\rangle$ into $C_{keep}$\;
    }

    \If{$A_{eq}\neq\emptyset$ and $I_{match}\neq\emptyset$}{

        $e'\leftarrow e\oplus\textsc{EqCategoricalMask}(p_{eq})$\;

        Add $\langle I_{match},e',A_{eq}\rangle$ into $C_{add}$\; 
    }
}

$C'\leftarrow C_{keep}\cup C_{add}$\;\label{line:merge_categorical_context}

\Return{$C'$}\;

\end{algorithm}

Algorithm~\ref{alg:update_categorical_context} presents the refinement process for categorical attribute groups. We define $EqCategoricalMask(p_{eq})$ as the evidence update mask associated with the current equality predicate, which is used to update the evidence accordingly. Based on the value of $t_i$ on the current attribute, the algorithm retrieves the candidate tuple set $I_{eq}$ with the same attribute value. Lines \ref{line:nullI_eq}-\ref{line:handleI_eq} handle the case when $I_{eq}$ is empty. $T.advance$ denotes the operation that advances from the current DCTrie nodes and returns the set of reachable nodes. When $I_{eq}$ is not empty, the algorithm processes each context separately. It calculates $I_{match}=I \cap I_{eq}$ and $I_{rem}=I\backslash I_{match}$. The original context is then split into two sub-contexts accordingly. Line \ref{line:merge_categorical_context} merges the new context to obtain the updated context set $C^\prime$, which serves as input for the next attribute group.
 
\begin{algorithm}[!t]
\caption{\textsc{UpdateNumericalContext}}
\label{alg:update_numerical_context}

\Input{The current context set $C$, inserted tuple $t_i$, numerical predicate group $g$, and DCTrie $T$}

\Output{Updated context set $C'$}

$p_{eq},p_{lt},p_{gt}\leftarrow$ the equality, less-than, and greater-than predicates in $g$\;

$I_{eq}\leftarrow$ the tuple ID set with the same value as $t_i$ on the attribute of $g$\;

$I_{lt}\leftarrow$ the tuple ID set with a greater value than $t_i$ on the attribute of $g$\;

Initialize $C_{add}\leftarrow\emptyset$ and $C_{keep}\leftarrow\emptyset$\;

\ForEach{$\langle I,e,A\rangle\in C$}{

    $A_{eq}\leftarrow T.\textsc{Advance}(A,p_{eq})$\;

    $A_{lt}\leftarrow T.\textsc{Advance}(A,p_{lt})$\;

    $I_{eq}^{c}\leftarrow I\cap I_{eq}$\;

    \If{$I_{eq}^{c}\neq\emptyset$ and $A_{eq}\neq\emptyset$}{

        $e_{eq}\leftarrow e\oplus\textsc{EqRangeMask}(p_{eq})$\;

        Add $\langle I_{eq}^{c},e_{eq},A_{eq}\rangle$ into $C_{add}$\;
    }

    $I\leftarrow I\setminus I_{eq}^{c}$\;

    \If{$I=\emptyset$}{

        \Continue\;
    }

    $I_{lt}^{c}\leftarrow I\cap I_{lt}$\;

    $I\leftarrow I\setminus I_{lt}^{c}$\;

    \If{$I_{lt}^{c}\neq\emptyset$ and $A_{lt}\neq\emptyset$}{

        $e_{lt}\leftarrow e\oplus\textsc{LtRangeMask}(p_{lt})$\;

        Add $\langle I_{lt}^{c},e_{lt},A_{lt}\rangle$ into $C_{add}$\;
    }

    \If{$I\neq\emptyset$}{

        $A_{gt}\leftarrow T.\textsc{Advance}(A,p_{gt})$\;

        \If{$A_{gt}\neq\emptyset$}{

            Add $\langle I,e,A_{gt}\rangle$ into $C_{keep}$\;
        }
    }
}

$C'\leftarrow C_{keep}\cup C_{add}$\;

\Return{$C'$}\;

\end{algorithm}

Algorithm~\ref{alg:update_numerical_context} presents the refinement process for numerical attribute groups. For a numerical group, we define a set of predicate masks to replace the default state of the current evidence with the true predicate. $EqRangeMask(p_{eq})$ updates the evidence for the equality branch. $LtRangeMask(p_{lt})$ updates it for the less-than branch. For each context, the algorithm first extracts the subset of tuples that satisfy the equality relation and propagates them along the equality predicate branch. If the corresponding active node set is not empty, it uses $EqRangeMask(p_{eq})$ to correct the evidence and generate a new sub-context of the equality direction. These tuples are then removed from the current right tuples. Subsequently, the algorithm extracts subsets of the remaining tuples that satisfy the less-than relationship and propagates along the less-than branch. Finally, the algorithm combines all the newly generated equalities, less-than branch contexts, and the retained default branch contexts to obtain the updated context set $C^\prime$.

We now explain why symmetric DCs need to be inserted. For DCs that only contain categorical attributes, their symmetric form is equivalent to the original form, so whether to insert it has no impact on the result. To reduce memory space, we only initialize the tuple pairs in the $(t_{new}, t_{old})$ and $(t_{new}, t_{new}^\prime)$ directions, and then symmetrize the obtained evidence, resulting in the tuple pairs corresponding to $(t_{old}, t_{new})$ and $(t_{new}^\prime, t_{new})$. For DCs that contain numerical attributes, if only a single direction of the path is retained in the DCTrie, information in the opposite direction may be lost during refinement. This can cause some contexts that should continue to propagate to be pruned incorrectly. For example, when $(t_7, t_3)$ is refined to the $D_7$ node, it should continue to move in the direction of $t_7.Hired > t_3.Hired$. If the green node is not added, the set of forward nodes will be empty. Therefore, this evidence context will be incorrectly pruned.

\begin{figure}
\centering
\includegraphics[width=3.3in]{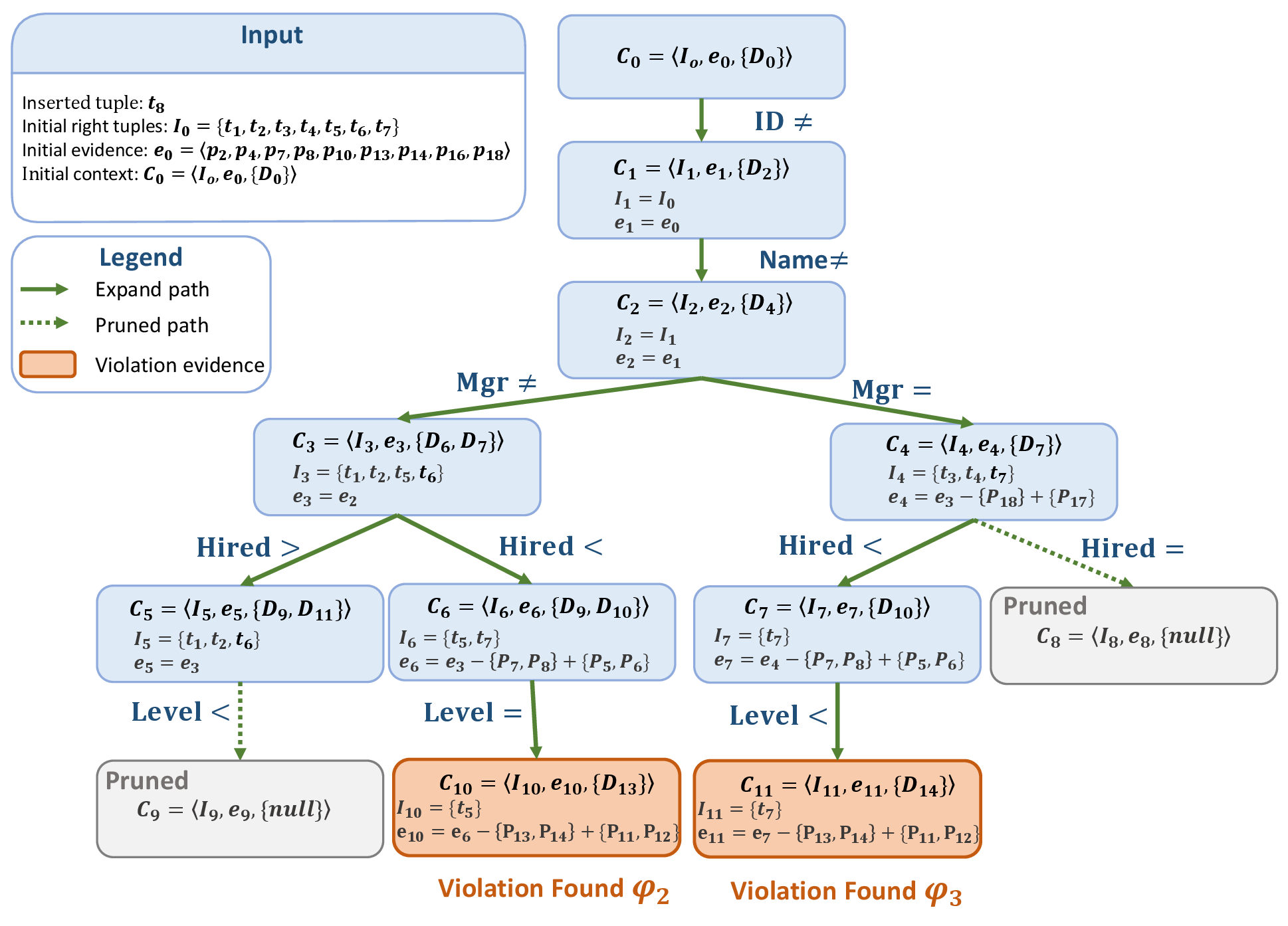}
\caption{Process of evidence refinement and pruning (taking $t_8$ as an example)}
\label{fig:pruning}
\end{figure}

Figure~\ref{fig:pruning} illustrates the evidence construction process using the insertion of tuple $t_8$ as an example. Initially, all tuple pairs are assumed to satisfy the default evidence and are associated with the root node. Since the $ID$ and $Name$ values of $t_8$ differ from those of all existing tuples, no context splitting occurs in these two layers. At the $Mgr$ layer, the context is divided into two sub-contexts. Tuples ${t_1,t_2,t_5,t_6}$ enter nodes ${D_6,D_7}$ through the inequality branch, while tuples ${t_3,t_4,t_7}$ enter node $D_7$ through the equality branch. At the $Hired$ layer, the context containing ${t_3,t_4}$ is pruned because no child node exists for the predicate $t.Hired=s.Hired$. Meanwhile, the inequality branch is further divided into ${t_1,t_2,t_6}$ and ${t_5}$. At the $Level$ layer, the context containing $t_5$ reaches leaf node $D_{13}$, identifying a violation of $\varphi_2$. Similarly, the context containing $t_7$ reaches leaf node $D_{14}$, identifying a violation of $\varphi_3$. Therefore, the violating tuple pairs $(t_8,t_5)$ and $(t_8,t_7)$ are successfully detected.

\subsection{DC Update Based on Incremental Evidence}
After incremental evidence construction, the algorithm records invalid DCs together with their corresponding violating evidence. The next step is to repair invalid DCs based on the detected violations. Similar to IncDC and 3DC, EviDC updates invalid DCs by adding more specific predicate restrictions.

The repaired DCs are generated through single-predicate extensions of invalid DCs. For example, the $(t_6, t_1)$ tuple pair violates $\varphi_1$. This can be achieved by rewriting $\varphi_1$ so that it is no longer violated by this violation pair. Adding more predicate restrictions on $\varphi_1$, such as $\varphi_1^\prime: \lnot(t.ID = s.ID \land t.Name \neq s.Name)$, indicates that there are no employees with the same name and ID. 

The DC update algorithm presented in EviDC breaks away from the double loop commonly found in DC verification. In previous methods, evidence records are first generated independently of the DCs, and each generated evidence record must then be matched against the current DCs to locate the violation. In contrast, our method records the invalid DCs and violating evidence uniformly in an object. During evidence context refinement, once the DCTrie search reaches a leaf node, it indicates that an invalid DC has been found. To quickly locate invalid DCs and their violating evidence, we extract the DC information recorded in the leaf node, create an object for each invalid DC, and store the corresponding evidence in it.

For an invalid DC, each predicate that does not appear in its violating evidence is treated as a candidate extension predicate. For a given DC $\varphi$, the algorithm first extracts the corresponding subset of violation evidence $E_\varphi$, which is the set of all predicates included in the corresponding violation evidence. Then, instead of enumerating candidate constraints aimlessly from the entire predicate space, predicates that do not appear in $E_\varphi$ are selected from the predicate space as candidate extension predicates. A candidate predicate is appended to the invalid DC to generate a repaired candidate. If the resulting DC passes the validity and non-triviality checks, it is retained.

It is worth noting that EviDC adopts a single predicate repair strategy. Although some violations can only be eliminated by adding multiple predicates simultaneously, enumerating all such combinations would result in a combinatorial explosion of candidate DCs. In the extreme case, a violation can always be eliminated by conjunctively adding all predicates from the predicate space to an invalid DC. Obviously, this is not a practical solution. Therefore, EviDC adopts a single predicate repair strategy. This design is consistent with existing incremental DC discovery methods such as IncDC and 3DC.

\subsection{Correctness and Complexity Analysis}
\subsubsection{Correctness}

The invariants maintained during evidence construction are evidence contexts. For any evidence context $c=<t,I,e,A>$, let $e$ denote the set of predicates that have already been verified as satisfied. The set $I$ always contains exactly the right-hand tuples currently satisfying evidence $e$. The set $A$ contains exactly the active nodes that may still be extended to a complete path in the DCTrie.

\begin{theorem}\label{theorem1}
Algorithm~\ref{alg:incremental_evidence_construction}
can identify all invalid DCs caused by inserted data.
\end{theorem}
\begin{IEEEproof}
Assume that $\varphi \in \Sigma$ is an invalid DC caused by inserted data. There must exist a tuple pair $(t_i,t_j)$ that satisfies all predicates along the path corresponding to $\varphi$, thereby violating $\varphi$ after the data update.

Initially, for each inserted tuple $t_i$, Algorithm~\ref{alg:incremental_evidence_construction} constructs an initial evidence context $c_0=\langle t_i,I_0,e_0,A_0\rangle$, where $I_0$ contains all tuples except $t_i$ and $A_0=\{\text{root}\}$. Therefore, tuple $t_j$ must belong to the candidate right-hand tuple set at the beginning.

At any level of the loop, the active node set of the evidence context containing $t_j$ is not empty and will not be pruned. Assume that at the $i$-th level, the active node set of the context containing $t_j$ becomes empty. The currently processed attribute predicate is $p_i$. It indicates that no path exists in the constructed DCTrie along the direction of $p_i$. However, this contradicts the fact that $\varphi$ has already been inserted into the DCTrie. Therefore, the evidence context containing $t_j$ must be preserved.

As a result, the context containing $t_j$ continues to survive throughout the subsequent refinement stages and advances along the DCTrie path corresponding to $\varphi$.

After all attribute groups have been processed, the refinement reaches the leaf node associated with $\varphi$. At this point, Algorithm~\ref{alg:incremental_evidence_construction} can correctly identify $\varphi$ as an invalid DC and does not miss any violating evidence.
\end{IEEEproof}

\begin{theorem}\label{theorem2}
The DC repair procedure of EviDC is sound.
\end{theorem}
\begin{IEEEproof}
We prove by contradiction. Suppose $\varphi$ is an invalid DC identified by EviDC. Let $E_{\varphi}$ denote the set of violating evidence corresponding to $\varphi$. A repaired DC derived from $\varphi$ has the form $\varphi'=\varphi\cup\{p\}$, where $p \in P_{add}$ and $P_{add}=P\setminus \bigcup_{e\in E_{\varphi}} e$.

Suppose that EviDC retains a repaired DC $\varphi'$ but $\varphi'$ is not valid on the updated dataset. Then there exists an evidence $e$ containing all predicates in $\varphi'$. Since $\varphi'$ is generated by adding a predicate $p\in P_{add}$ to $\varphi$, and $P_{add}$ only contains predicates absent from the violating evidence, we have $p\notin e$. However, because $p\in\varphi'$ and all predicates of $\varphi'$ are contained in $e$, it follows that $p\in e$, which is a contradiction. Therefore, every repaired DC retained by EviDC is valid on the updated dataset.

Hence the theorem holds.
\end{IEEEproof}

\begin{theorem}\label{theorem3}
EviDC terminates after a finite number of steps.
\end{theorem}
\begin{IEEEproof}
The number of inserted tuples is finite. Since the predicate space and the set of predicate groups are finite, the depth of DCTrie is also finite. For each inserted tuple, evidence refinement is performed over a finite number of predicate groups. Hence, the evidence construction stage performs a finite number of refinement operations.

For DC update, the number of invalid DCs identified during evidence construction is finite. Therefore, the repair stage also performs a finite number of operations.

Since both incremental evidence construction and DC update are finite processes, EviDC terminates after a finite number of steps.

\end{IEEEproof}

\subsubsection{Complexity}
Assume that the size of the original dataset is $n$, the number of inserted tuples is $|\Delta r|=m$, the number of predicates is $|P|$, and the number of predicate groups is $g$.

In the DCTrie construction phase, the algorithm traverses each DC in the original DC set and inserts it into the tree according to the predicate group order. For each DC, the insertion process accesses at most $g$ layers. Therefore, the time complexity of DCTrie construction is $O(|\Sigma|g)$. In terms of space, the number of nodes in DCTrie depends on the degree of prefix sharing among DCs. In the worst case, if there is no prefix sharing, each DC occupies an independent path of length at most $g$. Thus, the space complexity of DCTrie is $O(|T|)$, where $|T|\leq O(|\Sigma| \cdot g)$.

In the incremental evidence construction phase, each inserted tuple $t_i$ is compared with all tuples in $r\cup(\Delta r\setminus\{t_i\})$. Therefore, the total number of tuple pairs considered is $m(n+m-1)$, which is bounded by $O(m(n+m))$. For each tuple pair, evidence refinement proceeds through at most $g$ predicate groups, corresponding to the depth of DCTrie. Therefore, the worst-case time complexity of incremental evidence construction is $O(m(n+m)g)$.

In practice, EviDC does not retain all candidate tuple pairs. The DCTrie-guided pruning continuously reduces the number of candidate tuples and active paths during refinement. Let $n_i$ denote the average number of candidate tuples retained at level $i$, and let $a_i$ denote the average number of active nodes at that level. Then the practical complexity can be expressed as
\[
O\left(m\sum_{i=1}^{g} n_i a_i\right),
\]

where usually $n_i \ll n+m$ due to progressive context refinement and pruning.

In the DC update phase, EviDC repairs invalid DCs using single-predicate extensions. For each invalid DC $\varphi\in\Sigma^{-}$, let $U_{\varphi}=\bigcup_{e\in E_{\varphi}}e$ denote the set of predicates appearing in its violating evidence. The number of candidate predicates is $|P|-|U_{\varphi}|$. For each generated candidate DC, the algorithm performs validity and non-triviality checking. Let $C_{check}$ denote the cost of checking a candidate DC. Therefore, the update cost is
\[O\left(\sum_{\varphi\in\Sigma^{-}}(|P|-|U_{\varphi}|)\cdot C_{check}\right).\]

In the worst case, this can be bounded by
\[
O(|\Sigma^{-}|\cdot |P|\cdot C_{check}).
\]

Overall, the worst-case time complexity of EviDC is
\[O\!\left(|\Sigma| \cdot g+m(n+m)g+|\Sigma^{-}|\cdot |P|\cdot C_{check}\right).\]

\section{Experimental Evaluation}\label{section:experiment}
In this section, we conduct experiments to evaluate the effectiveness and scalability of EviDC. The results show that DCTrie reduces the computational overhead during DC updating in most tested settings.

\subsection{Settings}

All experiments were conducted on the same hardware platform equipped with an Intel Core i7-12700 processor running at 2.10 GHz, 12 physical cores, 32 GB RAM, and Windows 11. All algorithms were implemented or executed in Java using JDK 21. For each experiment, we repeated the execution 10 times and reported the average running time to reduce the influence of runtime fluctuations.

We evaluated the algorithms on both real-world and synthetic datasets that have been widely used in previous DC discovery studies~\cite{bleifuss2017efficient,livshits2020approximate,pena2019discovery,pena2022fast,qian2023incremental,pena2024discovering}. The synthetic Tax dataset is derived from~\cite{bohannon2006conditional} and contains representative attributes of tax records.

For each dataset, we first discovered initial DCs from the original relation. Then, we constructed incremental datasets with controlled violations. Specifically, violation tuples were generated according to the initial DC set and mixed with regular tuples to simulate inconsistent updates in dynamic data scenarios.

We implemented a naive exhaustive verifier to validate the correctness of EviDC. The verifier enumerates all evidence and checks every DC against all generated evidence to determine whether it is violated. Precision and Recall are computed based on the exact matching between the invalid DC sets returned by EviDC and the exhaustive verifier, rather than on the number of detected DCs. As shown in Table~\ref{tab:invalid_dc_correctness}, EviDC achieves perfect Precision and Recall on all tested datasets, confirming the consistency of its invalid DC detection results with those of the exhaustive verifier.

\begin{table}[t]
\centering
\caption{Correctness of Invalid DC Detection.}
\label{tab:invalid_dc_correctness}
\begin{tabular}{lcccc}
\toprule
Dataset & EviDC & Brute-force Verifier &Precision &Recall\\
\midrule
FD15     & 414  & 414 &1.0&1.0 \\
Hospital & 50   & 50   &1.0&1.0 \\
Claim    & 7    & 7    &1.0&1.0 \\
Atom     & 1817 & 1817 &1.0&1.0 \\
\bottomrule
\end{tabular}
\end{table}

We compare EviDC with two representative incremental DC discovery methods, namely 3DC and IncDC. Since the source code of 3DC is not publicly available, we implemented 3DC based on the algorithm description provided in~\cite{pena2024discovering} and the ECP framework~\cite{pena2022fast} released by the authors. We denote this implementation as 3DC* in the experimental results. To validate the reimplementation, we compared the behavior of 3DC* with the results reported in the original paper. The overall trends with respect to insertion ratios, evidence size, and runtime are consistent with those reported by the authors. IncDC was evaluated using the implementation publicly released by its authors.

For all methods, runtime measures the cost of a single incremental maintenance operation, including incremental evidence construction, invalid DC detection and update, and the maintenance of auxiliary structures required during the update. In the preprocessing stage, 3DC* constructs the original evidence set $E_r$, IncDC builds indexes over the initial dataset, and EviDC constructs the DCTrie. These preprocessing steps incur one-time costs and are therefore excluded from runtime measurements. Their construction times are reported separately in Table~\ref{tab:runtime_preprocessing}.

\subsection{Experimental Results}
\subsubsection{\bf{Exp-1: EviDC against 3DC, IncDC}}
To evaluate the performance of the proposed method under different update workloads, we compare the running times of the three algorithms on each dataset with insertion ratios $\lambda$ of 0.1\%, 1\%, 10\%, and 30\%. The results are reported in Table~\ref{tab:runtime_comparison}, where $OOM$ indicates that the algorithm failed to complete within the time limit.

\begin{table*}[!t]
\centering
\caption{Runtime comparison under different insertion ratios.}
\label{tab:runtime_comparison}
\begin{tabular}{lrrBccBccBccBcc}
\toprule
\multirow{2}{*}{Dataset} 
& \multirow{2}{*}{$|r|$}
& \multirow{2}{*}{$|R|$}
& \multicolumn{3}{c}{0.1\%} 
& \multicolumn{3}{c}{1\%} 
& \multicolumn{3}{c}{10\%} 
& \multicolumn{3}{c}{30\%} \\
\cmidrule(lr){4-6} \cmidrule(lr){7-9} \cmidrule(lr){10-12} \cmidrule(lr){13-15}
& & 
& EviDC & 3DC* & IncDC 
& EviDC & 3DC* & IncDC 
& EviDC & 3DC* & IncDC 
& EviDC & 3DC* & IncDC \\
\midrule
Tax      & 99904  & 15 & 0.04 & 0.08 & 6.07 & 0.13 & 0.21 & 56.83 & 0.61 & 1.34 & OOM    & 2.11 & 4.67 & OOM \\
Airport  & 55100  & 11 & 0.11 & 0.18 & 21.56 & 0.73 & 1.85 & 128.75 & 5.09 & 7.13 & OOM    & 12.91 & 18.12 & OOM \\
Hospital & 114920 & 15 & 0.02 & 0.04 & 0.09 & 0.05 & 0.05 & 0.34 & 0.15 & 0.26 & 10.69 & 0.32 & 0.78 & OOM \\
Ncvoter  & 675000 & 15 & 0.21 & 0.66 & 1086.24 & 1.46 & 4.34 & OOM & 3.65 & 40.30 & OOM & 48.87 & 145.01 & OOM \\
Atom     & 147067 & 13 & 0.05 & 0.06 & 14.43 & 0.13 & 0.11 & 235.47 & 0.77 & 1.75 & OOM & 2.21 & 5.14 & OOM \\
Flights  & 499300 & 17 & 0.70 & 0.66 & OOM & 5.85 & 5.74 & OOM & 49.31 & 53.09 & OOM & 167.84 & 184.21 & OOM \\
Adult    & 32561  & 15 & 0.12 & 0.24 & 33.76 & 0.77 & 1.06 & 84.45 & 7.63 & 8.91 & 352.17 & 18.80 & 30.37 & 1,043.94 \\
Claim    & 154061 & 11 & 0.02 & 0.19 & 0.14 & 0.04 & 0.74 & 0.61 & 0.09 & 1.93 & 4.37 & 0.23 & 4.55 & 18.97 \\
UCE      & 14246  & 11 & 0.91 & 2.39 & 1,427.05 & 5.74 & 8.01 & OOM & 42.34 & 66.63 & OOM & 97.91 & 269.27 & OOM \\
Dit      & 780000 & 11 & 0.45 & 1.89 & OOM & 2.86 & 4.50 & OOM & 28.47 & 38.52 & OOM & 91.04 & 132.99 & OOM \\
FD15     & 187500 & 15 & 0.08 & 0.12 & 0.98 & 0.15 & 0.17 & 1.06 & 0.26 & 0.39 & 1.39 & 0.45 & 0.73 & 1.76 \\
\bottomrule
\end{tabular}
\end{table*}

As shown in Table~\ref{tab:runtime_comparison}, the running times of all three methods increase with the insertion ratio. EviDC grows more slowly on most datasets. This is because its evidence construction is bounded by reachable violation paths rather than by the full predicate space. During refinement, its early termination reduces predicate comparison cost and context splitting cost.

Compared with 3DC*, EviDC reduces runtime by approximately 50\% in most scenarios. For example, at an insertion ratio of 30\%, EviDC takes 97 seconds on UCE, while 3DC* requires 269 seconds. The improvement comes from avoiding the construction and matching of evidence that cannot reach any DCTrie leaf. EviDC prunes many tuple pairs during construction and records an invalid DC when a DCTrie leaf is reached, reducing both intermediate evidence size and evidence-to-DC matching cost.

The gap becomes more evident at high insertion ratios. As more tuples are inserted, the number of affected tuple pairs increases rapidly, forcing 3DC* to construct and compare a larger amount of incremental evidence. In contrast, EviDC only keeps evidence contexts that can still reach some DCTrie path, and many irrelevant tuple pairs are pruned before complete evidence records are generated. Therefore, the saved predicate evaluation and evidence matching costs become larger as the insertion ratio increases.

At low insertion ratios, EviDC and 3DC* perform similarly on some datasets. On the Flights dataset, EviDC is slightly slower than 3DC* at low insertion ratios. This is because the initial DC set is relatively large, containing 57,340 DCs. A larger DC set leads to a larger DCTrie and increases the number of active nodes maintained during context refinement. When only a few tuples are inserted, the pruning benefit is not yet large enough to offset the trie traversal overhead. As the insertion ratio increases, more irrelevant tuple-pair contexts are pruned early, and the saved predicate evaluation and evidence storage costs become larger than the additional DCTrie traversal cost.

\begin{table}[!t]
\centering
\caption{Execution Time Comparison on Preprocessing Steps.}
\label{tab:runtime_preprocessing}
\begin{tabular}{lccc}
\toprule
Dataset & EviDC (DCTrie) & 3DC* ($E_r$) & IncDC (Index) \\
\midrule
Tax      & 0.07  & 12.95  & 18.37 \\
Airport  & 0.16  & 12.08  & 211.88 \\
Hospital & 0.02  & 2.68   & 9.89 \\
Ncvoter  & 0.28  & 371.16 & 166.01 \\
Atom     & 0.08  & 6.45   & 52.92 \\
Flights  & 0.18  & 499.42 & 259.64 \\
Adult    & 0.35  & 95.51  & 33.91 \\
Claim    & 0.04  & 16.95  & 11.89 \\
UCE      & 7.36  & 96.83  & 113.54 \\
Dit      & 1.37  & 353.44 & -- \\
FD15     & 0.07  & 2.40   & 44.75 \\
\bottomrule
\end{tabular}
\end{table}

Table~\ref{tab:runtime_preprocessing} compares the preprocessing cost of the three methods. $E_r$ denotes the original evidence set. Across most datasets, DCTrie construction finishes within one second. Even on larger datasets such as Flights and Ncvoter, EviDC only takes 0.18s and 0.28s, while 3DC* and IncDC require hundreds of seconds.

The reason is that DCTrie construction depends mainly on the number of existing DCs and the number of predicate groups. It inserts each DC as a path and merges common prefixes. In contrast, 3DC* constructs the original evidence set over tuple pairs, and IncDC builds indexes over the original data. Therefore, their preprocessing costs grow more directly with data size. This explains why DCTrie introduces much lower preprocessing time and memory overhead.

\subsubsection{\bf{Exp-2: Evidence Size under Different Insertion Ratios}}
In order to analyze the impact of DCTrie on evidence size, we counted the number of evidence under different insertion ratios. 

\begin{figure*}[t]
\centering
\subfloat[]{\includegraphics[width=1.6in]{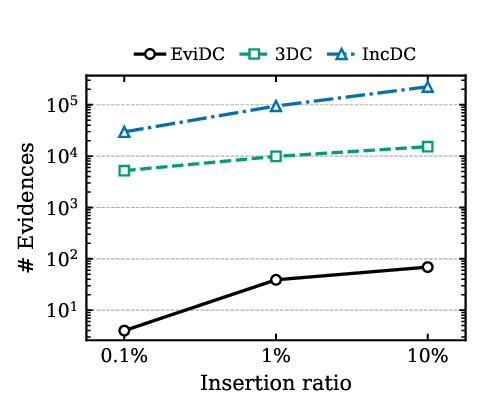}}
\subfloat[]{\includegraphics[width=1.6in]{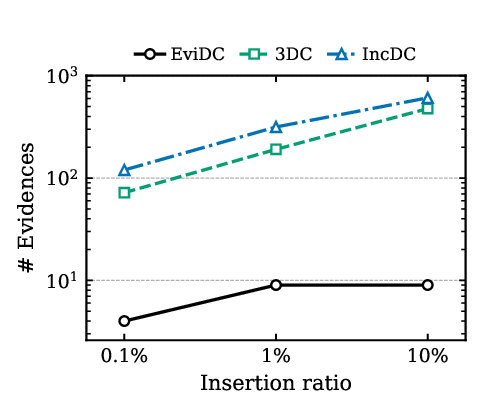}}
\subfloat[]{\includegraphics[width=1.6in]{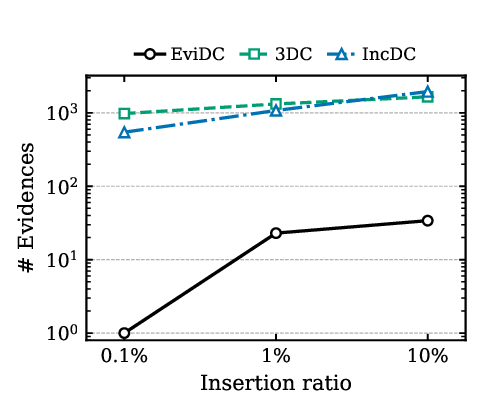}}
\subfloat[]{\includegraphics[width=1.6in]{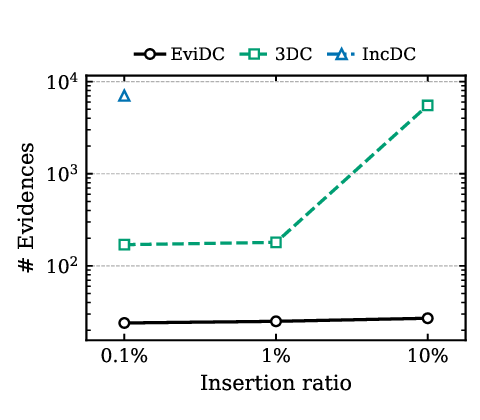}}

\vspace{-1mm}

\subfloat[]{\includegraphics[width=1.6in]{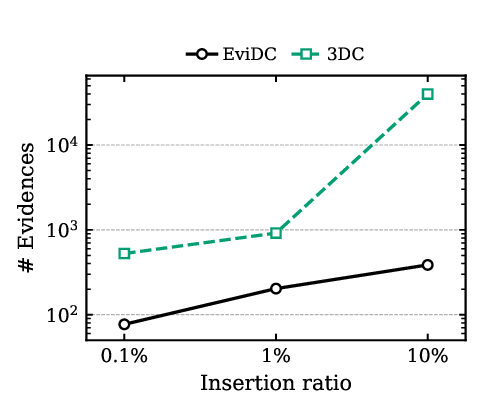}}
\subfloat[]{\includegraphics[width=1.6in]{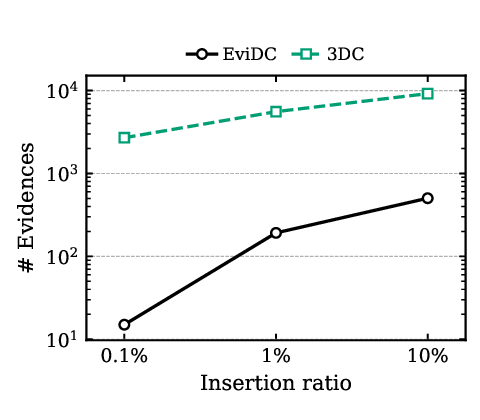}}
\subfloat[]{\includegraphics[width=1.6in]{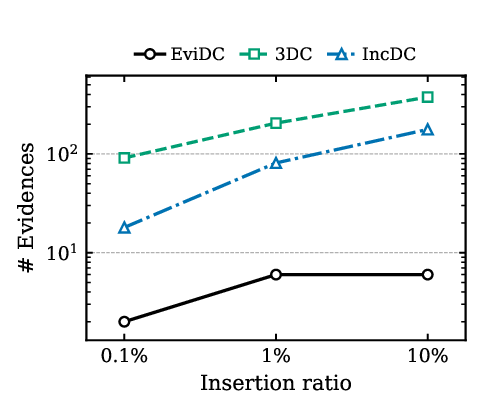}}
\subfloat[]{\includegraphics[width=1.6in]{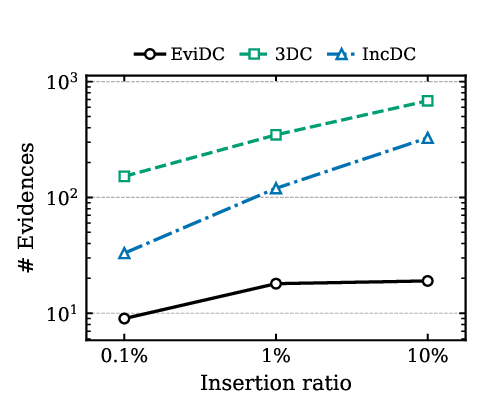}}

\vspace{-2mm}
\caption{Comparison of evidence size under different insertion ratios. (a) Airport, (b) Hospital, (c) Claim, (d) Tax, (e) Dit, (f) Flights, (g) FD15, (h) FD25.}
\label{fig:evidence_count}
\end{figure*}

In Figure~\ref{fig:evidence_count}, EviDC generates much less evidence than 3DC* and IncDC across most datasets and insertion ratios. This is because DCTrie performs reachability checking during evidence refinement. If partial evidence cannot be extended to any existing DC, EviDC stops refining the corresponding context and avoids materializing a complete evidence record. Therefore, the retained evidence mainly corresponds to potential violation patterns, while 3DC* and IncDC may still generate evidence for tuple pairs irrelevant to all existing DCs.

\subsubsection{\bf{Exp-3: Evidence Size under Different Data Sizes}}
We fixed the insertion ratio at 1\% and compared the changes in the amount of evidence generated by EviDC, 3DC, and IncDC under different original data sizes.
\begin{figure*}
\centering
\subfloat[]{\includegraphics[width=1.6in]{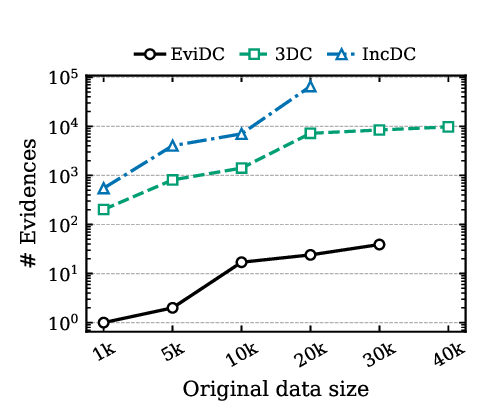}}
\subfloat[]{\includegraphics[width=1.6in]{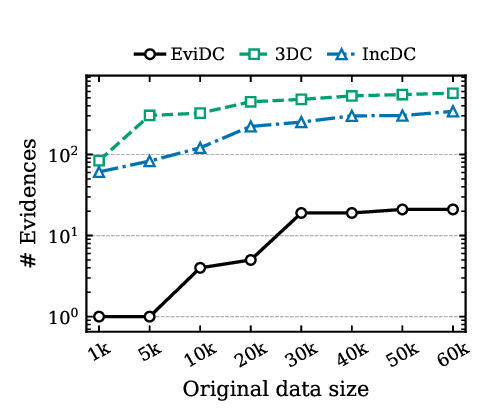}}
\subfloat[]{\includegraphics[width=1.6in]{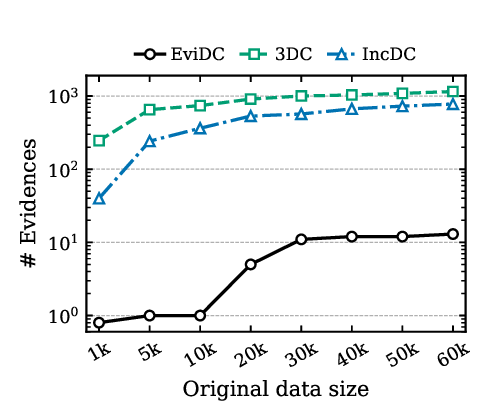}}
\subfloat[]{\includegraphics[width=1.6in]{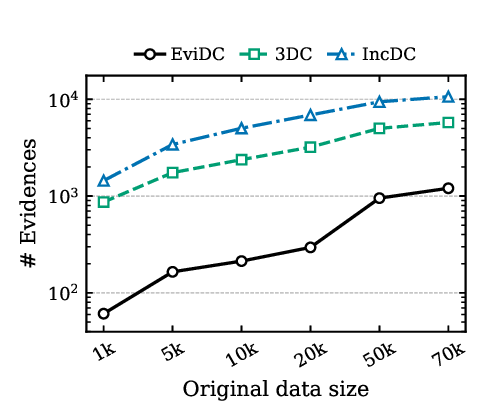}}

\vspace{-1mm}

\subfloat[]{\includegraphics[width=1.6in]{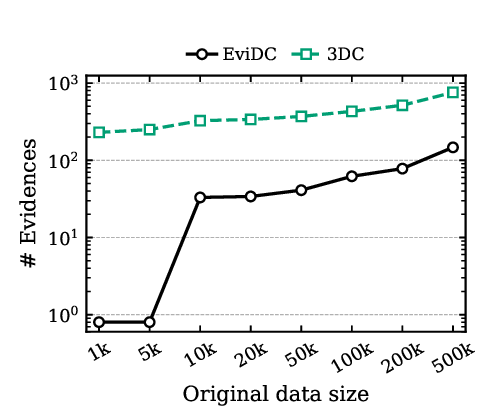}}
\subfloat[]{\includegraphics[width=1.6in]{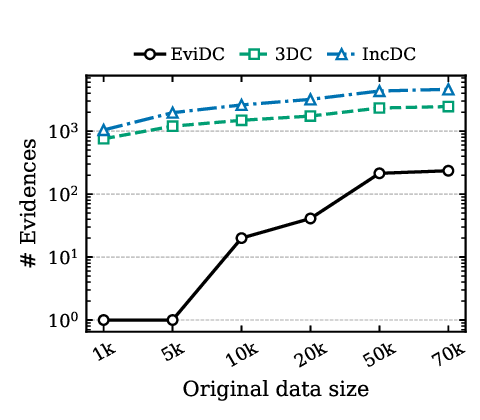}}
\subfloat[]{\includegraphics[width=1.6in]{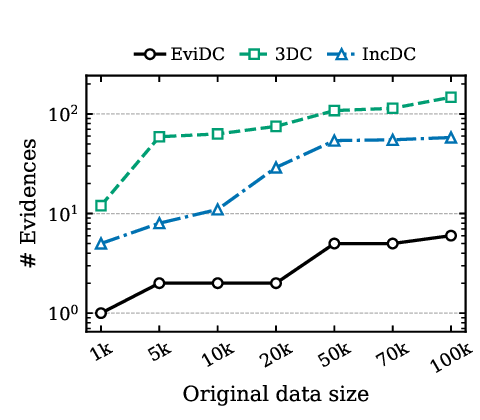}}
\subfloat[]{\includegraphics[width=1.6in]{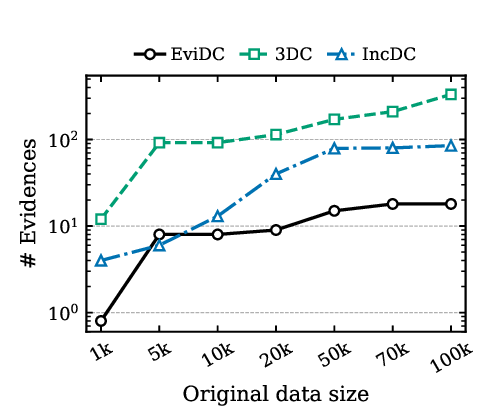}}
\caption{Comparison of evidence size under varying original data sizes. (a) Airport, (b) Hospital, (c) Claim, (d) Tax, (e) Dit, (f) Flights, (g) FD15, (h) FD25.}
\label{fig:EviDataSize}
\end{figure*}

In Figure~\ref{fig:EviDataSize}, all methods generate more evidence as the original data size increases, but EviDC grows more slowly. A clear gap between EviDC and the comparison methods can be observed on Flights, Hospital, and Claim. A larger dataset introduces more affected tuple pairs, but many of them cannot match any complete violation path in DCTrie. These tuple pairs are pruned after only a few predicate groups, reducing predicate comparison cost and intermediate evidence storage. This explains why EviDC suppresses evidence expansion more effectively as the data scale increases.

\subsubsection{\bf{Exp-4: Impact of Predicate Space Size}}
To evaluate the scalability of EviDC with respect to predicate space size and DC complexity, we use the FD dataset to generate four datasets, namely FD5, FD10, FD15, and FD20. These datasets contain 5, 10, 15, and 20 attributes. As the number of attributes increases, both the predicate space and the number of discovered DCs grow significantly. The results are shown in Figure~\ref{fig:fd_runtime}.

\begin{figure}
\centering
\subfloat[]{\includegraphics[width=1.6in]{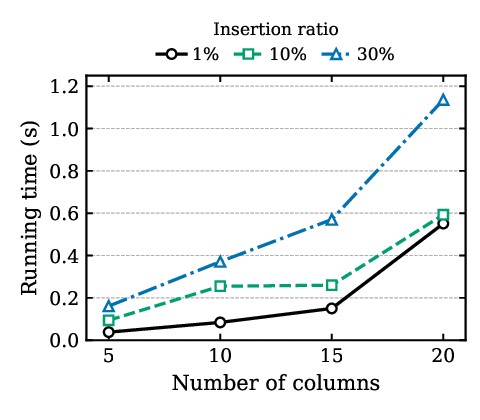}}%
\hfill
\subfloat[]{\includegraphics[width=1.6in]{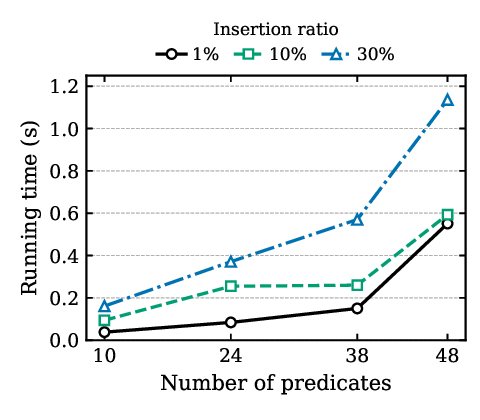}}%
\caption{Runtime comparison under different predicate and attribute counts. (a) runtime vs. attribute count, (b) runtime vs. predicate count.}
\label{fig:fd_runtime}
\end{figure}

Overall, the runtime of EviDC increases as the number of attributes, predicates, and original DCs grows. When the number of attributes increases from 5 to 20, the number of original DCs grows from 10 to 1120 and the number of predicates grows from 10 to 48. This trend is expected because a larger predicate space introduces more possible refinement branches, and a larger DC set leads to more potential violation paths in DCTrie. However, the increase in runtime is relatively moderate. 

This result shows that EviDC is not sensitive to the raw size of the candidate predicate space. The main reason is that EviDC does not enumerate evidence over the entire predicate space, but only expands branches reachable from the current active nodes in DCTrie. Therefore, even when the predicate space and the DC set become larger, many irrelevant branches are never expanded.

\subsubsection{\bf{Exp-5: Scalability with Increasing Data Size}}
To evaluate the scalability of EviDC on large datasets, we conduct experiments on the Tax dataset with data sizes ranging from 10k to 1M tuples. Three insertion ratios, namely 0.1\%, 1\%, and 10\%, are considered. The results are shown in Fig.~\ref{fig:tax}.

\begin{figure}
\centering
\includegraphics[width=2.5in]{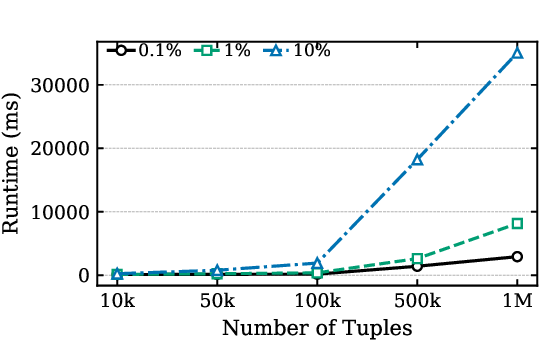}
\caption{Runtime of EviDC on Tax under different dataset sizes and insertion ratios.}
\label{fig:tax}
\end{figure}

As the dataset size increases, the runtime of EviDC exhibits an approximately linear growth trend under all insertion ratios. Even when the dataset reaches one million tuples, the algorithm remains efficient. This result indicates that, in this setting, the cost of EviDC grows primarily with the number of affected tuple pairs rather than the entire search space.

The impact of insertion ratio becomes more noticeable as the dataset grows. When the insertion ratio increases from 0.1\% to 10\%, the number of incremental tuple pairs increases significantly, leading to more evidence contexts and more DCTrie traversal operations. Nevertheless, the growth rate remains stable without a sudden performance degradation. This is because EviDC performs evidence construction only along reachable violation paths and prunes irrelevant contexts early. As a result, the algorithm reduces the intermediate structure expansion commonly observed in generate-then-verify approaches, which helps maintain scalability under increasing data sizes and larger update workloads in our experiments.

\subsubsection{\bf{Exp-6: Memory Usage Comparison}}

Figure~\ref{fig:memory_overview} shows the memory consumption of EviDC, 3DC, and IncDC across a range of datasets. The memory usage is measured in megabytes (MB) and reflects the peak memory consumption during the incremental DC discovery process.

\begin{figure}[!t]
\centering
\includegraphics[width=3.2in]{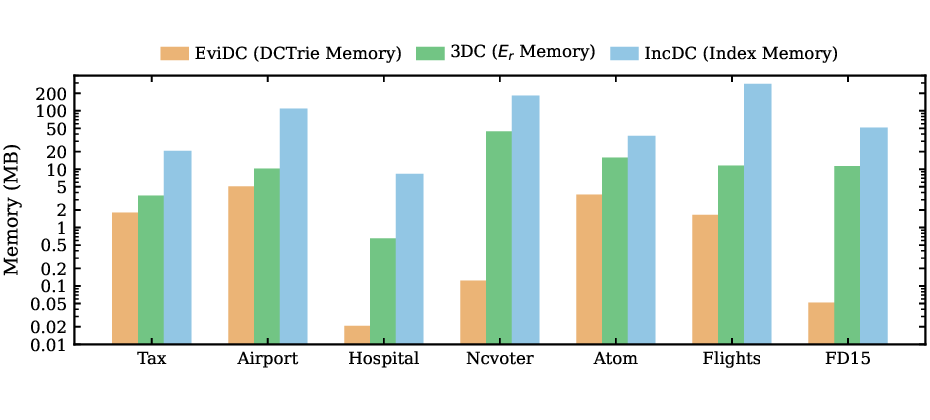}
\caption{Memory usage comparison of EviDC (DCTrie), 3DC ($E_r$), and IncDC (Index) on various datasets.}
\label{fig:memory_overview}
\end{figure}

As shown in Figure~\ref{fig:memory_overview}, EviDC achieves the lowest memory footprint across almost all datasets. This is because it does not maintain the complete historical evidence set required by 3DC* or the index structures used by IncDC. During incremental processing, EviDC stores only DCTrie and surviving evidence contexts, and contexts that cannot reach any violation path are removed immediately.

\subsubsection{\bf{Exp-7: Runtime Breakdown}}

To better understand the sources of EviDC's runtime cost, we further analyze the time distribution of different components in the algorithm. Specifically, the runtime is divided into five stages, including data reading, predicate space construction, DCTrie construction, incremental evidence construction, and DC update. Figure~\ref{fig:percentage} reports the percentage of total runtime spent in each stage on different datasets.

\begin{figure}
\centering
\includegraphics[width=3.4in]{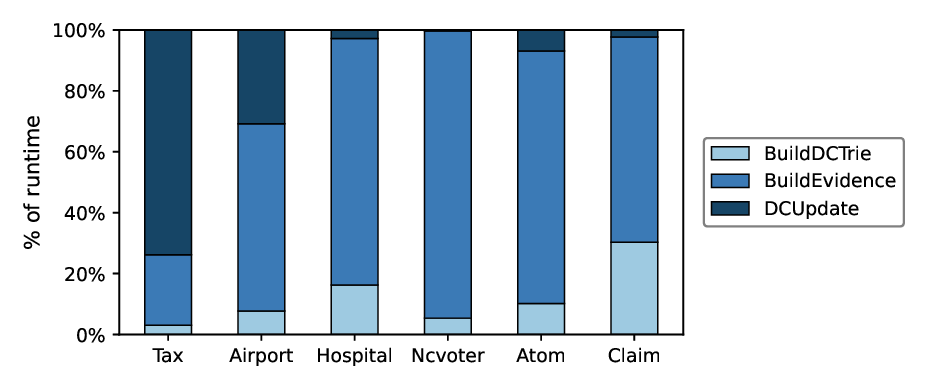}
\caption{Percentage breakdown of EviDC runtime on different datasets.}
\label{fig:percentage}
\end{figure}

Overall, the runtime mainly comes from incremental evidence construction, DCTrie construction, and data reading. Evidence construction dominates because it directly processes affected tuple pairs through predicate evaluation, context refinement, and active-node propagation. The DC update cost is relatively small because invalid DCs are located when DCTrie leaves are reached, avoiding an additional exhaustive matching step between generated evidence and existing DCs.

\subsection{Discussion}

The experimental results show that EviDC improves incremental DC discovery by using the structure of existing DCs during evidence construction. In this way, many irrelevant tuple pairs can be pruned early. This reduces both runtime and memory consumption. In our experiments, the construction of DCTrie and the predicate space only takes a small part of the total runtime.

However, EviDC is sensitive to the size and structure of the original DC set. If the DC set is very large, DCTrie construction and traversal can become more expensive. In this case, the advantage of EviDC may become smaller. When the insertion ratio grows, the benefit of EviDC becomes more evident. Therefore, EviDC is more suitable for scenarios with frequent insertions, where the DC structure can effectively guide evidence pruning. In the repair stage, EviDC adopts a single-predicate extension strategy, consistent with existing incremental DC discovery methods. This strategy aims to eliminate violations by adding more restrictions to each invalid DC, rather than exhaustively enumerating all possible repaired DCs.

\section{Conclusion} \label{section:conclusion}

This paper studies the problem of incremental denial constraint discovery. Existing methods construct incremental evidence without exploiting the structural information embedded in DCs, leading to a large amount of unnecessary intermediate results. To address this issue, we propose EviDC, a violation-guided incremental DC discovery method. EviDC organizes existing DCs as a prefix-sharing structure called DCTrie and guides evidence construction along potential violation paths, allowing irrelevant branches to be pruned early. Experiments show that EviDC significantly reduces the size of incremental evidence and achieves better efficiency in most scenarios. The advantage becomes more pronounced as the insertion ratio and dataset size increase. Future work will focus on supporting more general update operations, including tuple deletions and modifications.



\end{document}